\documentclass[11pt, letterpaper, reqno]{amsart}
\usepackage[utf8]{inputenc} 	
\usepackage{microtype} 			
\usepackage{geometry} 			
\usepackage{fancyhdr} 			
\usepackage{amsmath}			
\usepackage{amsthm}		 		
\usepackage{amssymb}	 		
\usepackage{esint}	 			
\usepackage{bm}					
\usepackage{mathrsfs}			
\usepackage{thmtools}
\usepackage{xcolor}				
\usepackage{tikz}				
\usepackage[all]{xy}			
\usepackage{graphicx}			
\usepackage{wrapfig}
\usepackage{enumitem} 			
\usepackage{booktabs}
\usepackage{array}
\usepackage{lmodern}			
\usepackage[T1]{fontenc}		
\usepackage[colorlinks=true, citecolor=cyan, urlcolor=magenta]{hyperref}
\usepackage{cleveref}
\usepackage[backend=biber, style=alphabetic, backref=false, firstinits=true, isbn=false, date=year]{biblatex}
\makeatletter
\renewbibmacro{in:}{}
\DeclareFieldFormat{pages}{#1}
\renewcommand*{\bibnamedash}{%
	\leavevmode\raise +0.6ex\hbox to 5.5ex{\hrulefill}.\space\space}

\InitializeBibliographyStyle{\global\undef\bbx@lasthash}

\newbibmacro*{bbx:savehash}{\savefield{fullhash}{\bbx@lasthash}}

\renewbibmacro*{author}{%
	\ifboolexpr{
		test \ifuseauthor
		and
		not test {\ifnameundef{author}}
	}
	{%
		\iffieldequals{fullhash}{\bbx@lasthash}
		{\bibnamedash\addcomma\space}
		{\printnames{author}}%
		\usebibmacro{bbx:savehash}%
		\iffieldundef{authortype}
		{}
		{%
			\setunit{\addcomma\space}%
			\usebibmacro{authorstrg}%
		}%
	}
	{\global\undef\bbx@lasthash}%
}
\makeatother
\newtheorem{propositionx}{Proposition}[section]
\newenvironment{proposition}
{\pushQED{\qed}\propositionx}
{\popQED\endpropositionx}
\newenvironment{propositionp}
{\pushQED{\qed}\propositionx}
{\popQED\endpropositionx}
\newtheorem{theoremx}[propositionx]{Theorem}
\newenvironment{theorem}
{\pushQED{\qed}\theoremx}
{\popQED\endtheoremx}

\theoremstyle{remark}
\newtheorem{remark}{Remark}[section]

\newcommand{\dd}{\,\mathrm{d}}

\makeatletter
\@namedef{subjclassname@2020}{\textup{2020} Mathematics Subject Classification}
\makeatother

\newcommand{\bbC}{\mathbb{C}}

\newcommand{\bbN}{\mathbb{N}}

\newcommand{\bbR}{\mathbb{R}}

\newcommand{\calD}{\mathcal{D}}

\newcommand{\calH}{\mathcal{H}}

\newcommand{\calL}{\mathcal{L}}

\newcommand{\calS}{\mathcal{S}}

\newcommand{\calV}{\mathcal{V}}

\title{Massive waves gravitationally bound to static bodies}
\author{Ethan Sussman}
\date{December 1st, 2023 (Last update). November 15, 2022 (Draft)}
\email{ethanws@stanford.edu}
\address{Department of Mathematics, Stanford University, California, USA}
\subjclass[2020]{Primary 35P05. Secondary 35P15, 81Q05.}

\begin{document}

\begin{abstract}
	We show that, given any static spacetime whose spatial slices are asymptotically Euclidean (or, more generally, asymptotically conic) manifolds modeled on the large end of the Schwarzschild exterior, there exist stationary solutions to the Klein--Gordon equation having Schwartz initial data. In fact, there exist infinitely many independent such solutions. The proof is a variational argument based on the long range nature of the effective potential. We give two sets of test functions which serve to verify the hypothesis of the variational argument. One set consists of cutoff versions of the hydrogen bound states and is used to prove the existence of eigenvalues near the hydrogen spectrum.
\end{abstract}

\maketitle

\tableofcontents

\section{Introduction}

In classical Newtonian gravity, massive particles can be bound to the gravitational potential-well generated by another body. 
Solutions to the Klein--Gordon equation 
\begin{equation}
	\square U + m^2 U = 0
\end{equation}
serve as wavefunctions for massive scalar particles in relativistic quantum mechanics, so it is to be expected that they can get gravitationally bound, in some suitable sense, to astrophysical bodies. 
Here, $\square$ is the d'Alembertian of the spacetime, with the sign chosen so that the spatial Laplace--Beltrami operator is positive semidefinite.
One manifestation of gravitational binding should be a lack of temporal decay, but this intuition should be taken with a grain of salt for at least two reasons:
\begin{itemize}
	\item in classical Newtonian gravity, the mass of a particle is irrelevant to its orbital motion, but solutions to the massless wave equation $\square U = 0$ (on astrophysical spacetimes) do actually decay, specifically at a $\sim t^{-3}$ rate, a fact known as \textit{Price's law} \cite{PriceI, PriceII}, and 
	\item it has been predicted by physicists that, on the exact Schwarzschild exterior and some of its relatives, solutions to the Klein--Gordon equation also decay, but at a different rate, namely $\sim t^{-5/6}$ \cite{HodTsvi}\cite{KoyamaTomimatsu, KoyamaTomimatsu2}\cite{BurkoKhanna}\cite{KonoplyaETAL}\cite{BarrancoETAL}.  
\end{itemize}
We consider in this note a broad class of static spacetimes whose asymptotic structure is given by the \textit{large end} of the Schwarzschild (or, more generally, Reissner--Nordstr\"om)  exterior. A precise definition appears below. One key example is any static spacetime whose spatial slices are isometric to the large end of the Schwarzschild exterior outside of some compact subset. The exact Schwarzschild exterior is excluded. 
This is because the Schwarzschild exterior has two ends -- the ``large'' end, where $r\to\infty$, and the horizon -- whereas the ``admissible'' metrics considered here only have one. 
Admissible metrics appear in nature as the gravitational field configurations generated by static astrophysical bodies lacking the necessary density to form a black hole. 
As such, they provide a model for the gravitational field of stars, planets, etc.\ in the limit where the angular momentum is negligible. 
Price's law applies to such spacetimes. In this generality, this has been proven rigorously by Hintz \cite{HintzPrice} --- see also \cite{DSS1, DSS2}\cite{TataruPrice}\cite{Tatarut3}\cite{AAG1,AAG2}\cite{Morgan, MorganWunsch}\cite{Looi}. 
On the other hand, confirming (or disconfirming) physicists' predictions regarding Klein--Gordon on exact Schwarzschild remains an open problem. In fact, proving even $o(1)$ decay remains an open problem. However, there has been very recent progress in the case when the initial data involves only finitely many spherical harmonics \cite{Pasqualotto}.

Our main goal is to prove that not even $o(1)$ decay applies to admissible spacetimes:
\begin{theorem}
	Let $\bbR_t\times X$ denote an admissible spacetime. Then, for each $m>0$, there exists an infinite sequence $\{E_n\}_{n=1}^\infty$ of $E_n \in (0,m^2)$ with $E_n\downarrow 0$ such that there exists, for each $n\in \bbN^+$, a Schwartz function $u_n:X\to \bbR$,
	not identically 0, such that 
	\begin{equation}
		U_n(t,-) = e^{i t\sqrt{m^2-E_n} }u_n
	\end{equation}
	satisfies $\square U_n + m^2 U_n = 0$. 
\end{theorem}

So, on any admissible spacetime, there exist temporally non-decaying solutions to the Klein--Gordon equation. 
This contrasts with the situation for massless waves, for which the asymptotic structure at infinity is intimately related to wave decay \cite{Morgan}\cite[\S4.3]{Hintz}. The decay of massive waves on the exact Schwarzschild exterior (assuming that such decay does in fact occur) is \textit{not} due to the asymptotic structure of the spacetime at the large end. The rough conjecture here would be that massive waves with insufficient kinetic energy do not radiate away from a black hole but rather fall towards the horizon. As the admissible spacetimes considered here look like the Schwarzschild exterior but lack a horizon, there is nowhere for the mass to go, and so solutions to the Klein--Gordon equation need not decay. 

In the body of the paper we will also consider the Klein--Gordon--Schr\"odinger equation, in which a short range potential has been added to the Klein--Gordon operator.

We start with the elementary observation that, given any stationary spacetime $(\bbR_t\times X,g)$, with $g$ constant in $t$, there exists a 1-parameter family 
\begin{equation} 
\{P(\sigma)\}_{\sigma \in \bbC} = \{P_m(\sigma)\}_{\sigma \in \bbC}\subset \operatorname{Diff}^2(X^\circ)
\end{equation} 
of 2nd order differential operators (depending on $m$, though we do not explicitly write this dependence below) on $X$ such that solutions $u \in \smash{\calD'(X)}$ to $P(\sigma)u=0$ yield non-decaying solutions $U$ to the Klein--Gordon equation. When the spacetime is not just stationary but actually static, in addition to asymptotically Schwarzschild (in which case $X$ is regarded as a manifold-with-boundary), then 
\begin{equation} 
	P(\sigma) = P +m^2- \sigma^2
	\label{eq:Ps}
\end{equation} 
is the spectral family of an $m$-dependent scaled Schr\"odinger operator $P=P_m$ with a potential of the form $V_1+V_2$, where $V_1 = - \mathsf{M} m^2/r$ and $V_2$ is a short range potential depending on $m$ and the metric of the spacetime. Thus, we have an \textit{attractive} Coulomb potential
proportional to the Schwarzschild mass $\mathsf{M}>0$ and the Klein--Gordon mass-squared $m^2$. 
This (except, perhaps, for the fact that it is $m^2$ rather than $m$ that shows up) should be unsurprising given the form of the potential in Newtonian gravity.
(We are working here in ``natural units'' with respect to which the Newtonian gravitational constant is given by $G=1/2$.) The low energy scattering theory of such operators was considered in \cite{Sussman} --- this corresponds to the $\sigma \to m^+$ limit. Here, we consider bound states with close to threshold energy, which instead involves the $\sigma \to m^-$ limit. 

The operator $P$, with the $L^2$-based Sobolev space $H^2(X)$ as a domain, is self-adjoint with respect to the inner-product of a carefully chosen $L^2$-space on $X$ (care required due to the rescaling in the definition of $P$), so the spectrum of $P$, defined accordingly, lies on the real axis.  
As is known by virtue of suitable elliptic theory,
\begin{equation}
	\sigma(P) = \{-E_n\}_{n=1}^N \cup [0,\infty),
\end{equation}
even if $\mathsf{M}=0$ or $\mathsf{M}<0$, where $N\in \bbN \cup \{\infty\}$ and may a priori be zero, and $E_1 > E_2\cdots > 0$ is a strictly decreasing sequence of positive real numbers whose only possible accumulation point is zero. Each $E_n$ is an eigenvalue of $P$, with a finite dimensional space of Schwartz eigenfunctions.  
One of very many ways to prove this is using Melrose's sc-calculus \cite{MelroseSC, MelroseGeometric}\cite{VasyGrenoble}, which, as an algebra, consists of the unital algebra $\operatorname{Diff}_{\mathrm{sc}}(X)$ of differential operators on $X^\circ$ generated over $C^\infty(X)$ by the vector fields $\rho V$, for $\rho$ a boundary-defining-function and $V$ a vector field tangent to the boundary.
Indeed, for $\lambda\in \bbC$ with $\lambda\notin [0,\infty)$, the differential operator $P - \lambda$ is elliptic in Melrose's sense, so analytic Fredholm theory applies there.
In this paper, we employ, when $\mathsf{M}>0$, a variational argument in order to show that the number of linearly independent bound states is infinite. So, $N=\infty$.

The proof of this theorem is contained in \S\ref{sec:main}, which is self-contained.
In \S\ref{sec:quasimodes}, we provide a more detailed investigation of the distribution of the eigenvalues of $P$.

On more general spacetimes than the static, horizon-free ones considered here, the family $\{P(\sigma)\}_{\sigma \in \bbC}$ is somewhat more complicated. For instance, on non-static stationary spacetimes $(\bbR_t\times X,g)$, with $g$ constant in $t$, 
\begin{equation}
	P(\sigma) = P  + i \sigma Q  + m^2 - \sigma^2
\end{equation}
for some first-order differential operator $Q\neq 0$ on $X^\circ$ with real coefficients. Thus, $P(\sigma)$ is no longer a spectral family, and the techniques below no longer apply. As indicated by \cite{Rothman}, the situation can be quite different. The presence of an event horizon complicates matters further, as it obstructs appeals to Fredholm theory (such as those below).
This is most easily illustrated on the exact Schwarzschild exterior, where the radial part $R(\sigma)$ of $P(\sigma)$ is 
\begin{equation}
	R(\sigma)= - \frac{\partial^2}{\partial r_*^2} - \frac{2}{r} \Big( 1 - \frac{\mathsf{M}}{r} \Big) \frac{\partial}{\partial r_*} + m^2 - \sigma^2  - \frac{\mathsf{M} m^2}{r} 
\end{equation}
with respect to the tortoise coordinate $r_*=r+ \mathsf{M} \log(\mathsf{M}^{-1} r-1)$. Since the second-order term is the Laplacian on $\bbR_{r_*}$, it makes sense to analyze this ordinary differential operator in $\operatorname{Diff}_{\mathrm{sc}}(\bbR_{r_*})$.
The large end of the spacetime corresponds to the $r_*\to \infty$ limit, where
\begin{equation}
	m^2 - \sigma^2  - \frac{\mathsf{M} m^2}{r}   = m^2 - \sigma^2 + O \Big( \frac{1}{r_*} \Big), 
\end{equation}
so $R(\sigma)$ is elliptic there, as an element of $\operatorname{Diff}_{\mathrm{sc}}(\bbR_{r_*})$, if $\sigma^2 < m^2$. 
The horizon corresponds to the $r_*\to -\infty$ limit, in which 
\begin{equation}
	 m^2 - \sigma^2  - \frac{\mathsf{M} m^2}{r}  =  - \sigma^2 + O (e^{-|r_*|/\mathsf{M}}), 
\end{equation}
so $R(\sigma)$ is \textit{not} elliptic there.

In fact, on the exact Schwarzschild exterior, $P$ has no bound states, as can be shown by an elementary calculation involving Wronskians for the radial ODE. A version of the variational argument still goes through, but rather than conclude the existence of infinitely bound states, we can only conclude that $\sigma(P)\cap (-\infty,0)$ is infinite. This is consistent with the continuous spectrum being $\sigma_{\mathrm{cont}}(P) = [-m^2,\infty)$ and the pure-point spectrum being empty.

\section{Variational argument}
\label{sec:main}

Fix $\delta \in (0,1]$. 
Consider a static Lorentzian spacetime of the form $(\bbR_t\times X,g)$, where $X$ is a compact $d\in \bbN^+$ dimensional manifold-with-boundary and $g$ is a Lorentzian metric of the form 
\begin{equation}
	g = -(1+x\aleph+x^{1+\delta} \beth)\cdot  \mathrm{d}t^2 + g_{X},
\end{equation}
where $\aleph\in \bbR$, $\beth \in S^0(X;\bbR)$, and $g_X$ is a (symbolic) asymptotically Riemannian conic metric on $X$ that is classical to subleading order and symmetric to subleading order in the radial-radial direction, i.e.\ a Riemannian metric of the form 
\begin{equation}
	g_X = \Big( \frac{1}{x^4} + \frac{\mathsf{M}}{x^3} \Big) \dd x^2 + \frac{h_{\partial X}}{x^2} + \frac{\Gamma_{1,\partial X} \odot \dd x}{x^2} + \frac{h_{1,\partial X}}{x} + x^{1+\delta} h_X
\end{equation}
with respect to some boundary collar $\iota:[0,\bar{x})_x \times \partial X \to X$, where $\bar{x}\in (0,\infty)$ and $x\in C^\infty(X;[0,\infty))$ denotes a boundary-defining function, and where the other terms are  
\begin{itemize}
	\item a Riemannian metric $h_{\partial X}$ on $\partial X$,
	\item a constant $\mathsf{M}\in \bbR$,
	\item a symbolic family of 1-forms $\Gamma_{1,\partial X} \in S^0([0,\bar{x})_x;\Omega^1(\partial X))$,
	\item a symbolic family of (not-necessarily positive semidefinite) symmetric 2-tensors $h_{1,\partial X}\in \smash{S^0([0,\bar{x})_x;C^\infty(\partial X;\operatorname{Sym}^2 T^* \partial X))}$,
\end{itemize}
and a symbolic remainder 
\begin{equation} 
	h_X \in S^0(\operatorname{Sym} {}^{\mathrm{sc}}T^* X ).
\end{equation} 
We say that the given spacetime is \emph{admissible} if, in addition to the requirements above, $\aleph<0$. 
We refer to \cite{MelroseSC}\cite{Sussman} for undefined notational conventions.

The condition that $g$ is Lorentzian means that $1+x \aleph + x^{1+\delta} \beth >0$ everywhere, so $(1+x \aleph + x^{1+\delta} \beth)^{\alpha}$ defines an element of $C^\infty(X;\bbR^+)$ for every $\alpha\in \bbR$. For the spacetimes of physical interest, 
\begin{equation} 
	\aleph = - \mathsf{M},
\end{equation} 
although we do not enforce this relation. 
For Reissner--Nordstr\"om-like metrics, $\delta=1$, and $\beth|_{\partial X}$ is constant, being related to the electric charge of the astrophysical body generating the gravitational field.

A straightforward calculation yields:
\begin{propositionp}
	The d'Alembertian $\square =- |g|^{-1/2} \sum_{j,k=0}^d \partial_j( |g|^{1/2}  g^{jk}\partial_k)$ has the form
	\begin{equation}
		\square =  \frac{1}{1+x\aleph + x^{1+\delta}\beth} \frac{\partial^2}{\partial t^2}  - \frac{1}{2} \frac{x^{1+\delta} \nabla \beth + (\aleph+(1+\delta) x^\delta\beth)\nabla x }{1+x \aleph + x^{1+\delta}\beth} + \triangle ,
	\end{equation}
	where $\triangle$ is the \emph{positive semidefinite} Laplace--Beltrami operator of the Riemannian manifold $(X,g_X)$, which we consider as an operator on $\bbR_t\times X$. Near $\partial X$, $\triangle$ has the form  
	\begin{equation}
		\triangle =  - \Big( 1 - \frac{\mathsf{M}}{r}\Big) \frac{\partial^2}{\partial r^2}+ \frac{1}{r^2} \triangle_{\partial X} - \frac{d-1}{r} \frac{\partial}{\partial r}  + \frac{1}{r}Q + S^{-1-\delta} \operatorname{Diff}_{\mathrm{sc}}^2(X;\bbR)
	\end{equation}
	with respect to the given boundary collar, where $r=1/x$, where $\triangle_{\partial X}$ is the (positive semidefinite) Laplace--Beltrami operator of $(\partial X,h_{\partial X})$ and $Q \in S^{0} \operatorname{Diff}_{\mathrm{sc}}^2(X;\bbR)$ has the form 
	\begin{equation}
		Q =  \frac{1}{r} Q_\perp \frac{\partial}{\partial r}+\frac{1}{r^2} Q_\partial  + \frac{1}{r} Q_{1,\partial}
	\end{equation}
	for $Q_\perp , Q_{1,\partial} \in S^0([0,\bar{x})_x ; \calV(\partial X;\bbR))$ and $Q_{\partial} \in S^0([0,\bar{x})_x ; \calV(\partial X;\bbR)^2\oplus \calV(\partial X;\bbR))$, where $\calV(\partial X;\bbR)$ is the space of vector fields on $\partial X$ with real coefficients.
	\label{prop:squareform}
\end{propositionp}

See \cite[Proposition 6.1]{Sussman} for details regarding the computation of $\triangle$. 

Note the absence of zeroth order terms in $Q_\perp,Q_{1,\partial},Q_\partial$, as such terms can be absorbed into the $S^{-1-\delta} \operatorname{Diff}_{\mathrm{sc}}^2(X;\bbR)$ error.

Fix $V\in S^{-1-\delta}(X;\bbR)$ and $m>0$. 
Consider the rescaled Schr\"odinger operator $P=P_m$ on $X^\circ$ given by 
\begin{equation}
	P =(1+x \aleph+x^{1+\delta} \beth) \triangle - \frac{1}{2} (x^{1+\delta} \nabla \beth + (\aleph +(1+\delta) x^\delta \beth) \nabla x)  +V_{\mathrm{eff}},
\end{equation}
where $V_{\mathrm{eff}} \in x \bbR +x^{1+\delta} S^0(X;\bbR)$ is given by $V_{\mathrm{eff}}=x m^2 \aleph+   x^{1+\delta} m^2\beth  + (1+x\aleph+x^{1+\delta} \beth)V$. Observe that $\nabla \beth \in S^{-1} \operatorname{Diff}_{\mathrm{sc}}^1(X)$ and $\nabla x \in x^2 S^0\operatorname{Diff}_{\mathrm{sc}}^1(X)$. 

At the level of sets, $L^2(X,\dd \mathrm{Vol}_{g_X})=L^2(X,(1+x\aleph+x^{1+\delta}\beth)^{-1/2} \dd \mathrm{Vol}_{g_X})$. We use `$\calS(X)$' to denote the set of Schwartz functions on $X$, and we abbreviate $\calH = L^2(X,(1+x\aleph+x^{1+\delta}\beth)^{-1/2} \dd \mathrm{Vol}_{g_X})$.

Let $P(E) = P+E$. Note that this parametrization convention differs from \cref{eq:Ps}.
\begin{proposition}
	$P:H^2(X)\to L^2(X)$ defines a lower-semibounded self-adjoint operator on $L^2(X,(1+x\aleph+x^{1+\delta}\beth)^{-1/2} \dd \mathrm{Vol}_{g_X})$.
\end{proposition}
\begin{proof}
	Let $\tilde{P} =  (1+x\aleph+x^{1+\delta}\beth)^{+1/4} P(1+x\aleph+x^{1+\delta}\beth)^{-1/4}$ denote the symbolic differential operator
	\begin{equation}
		\tilde{P} u = (1+x\aleph+x^{1+\delta}\beth)^{+1/4} P((1+x\aleph+x^{1+\delta}\beth)^{-1/4} u).
	\end{equation} 
This has the form 
\begin{equation} 
\tilde{P} = (1+x\aleph+x^{1+\delta}\beth) \triangle + W
\label{eq:W}
\end{equation} 
for some $W\in S^0(X;\bbR)$, so by the symmetry of $\triangle$ as a bilinear form on $L^2(X, \dd \mathrm{Vol}_{g_X})$, 
	\begin{equation}
		\int_X \frac{u_0^*   \tilde{P} v_0 \dd \mathrm{Vol}_{g_X} }{1+x\aleph + x^{1+\delta}\beth}= \int_X \frac{(\tilde{P}u_0)^*    v_0 \dd \mathrm{Vol}_{g_X}}{1+x\aleph + x^{1+\delta}\beth}
	\end{equation}
	for all $u_0,v_0\in \calS(X)$. Consequently, for all $u,v\in \calS(X)$, 
	\begin{align}
		\begin{split} 
		\int_X \frac{u^*   P v \dd \mathrm{Vol}_{g_X} }{(1+x\aleph + x^{1+\delta}\beth)^{1/2}} &= \int_X \frac{u^*     }{(1+x\aleph + x^{1+\delta}\beth)^{3/4}}  \tilde{P} \Big[ (1+x\aleph+x^{1+\delta}\beth)^{1/4}v \Big] \dd \mathrm{Vol}_{g_X} \\
		&= \int_X \Big[ (1+x\aleph+x^{1+\delta}\beth)^{1/4}u \Big]^*  \tilde{P} \Big[ (1+x\aleph+x^{1+\delta}\beth)^{1/4}v \Big]\;\;\: \frac{\dd \mathrm{Vol}_{g_X}}{1+x\aleph+x^{1+\delta}\beth} \\ 
		&= \int_X \Big[ (1+x\aleph+x^{1+\delta}\beth)^{1/4}v \Big]   \Big[\tilde{P}( (1+x\aleph+x^{1+\delta}\beth)^{1/4}u )\Big]^* \frac{\dd \mathrm{Vol}_{g_X}}{1+x\aleph+x^{1+\delta}\beth} \\
		&= \int_X\frac{v}{(1+x\aleph+x^{1+\delta}\beth)^{3/4}} \Big[\tilde{P}( (1+x\aleph+x^{1+\delta}\beth)^{1/4}u )\Big]^* \dd  \mathrm{Vol}_{g_X} \\
		&= \int_X\frac{(Pu)^* v \dd  \mathrm{Vol}_{g_X}}{(1+x\aleph+x^{1+\delta}\beth)^{1/2}},  
		\end{split}
	\end{align}
	which says that $P$ defines a symmetric bilinear form $(\calS(X)^2,\langle-,-\rangle_{\calH})\to \bbC$. The same computations show that 
	\begin{equation}
		\langle u, Pv \rangle_\calH  = \langle Pu,v \rangle_{\calH} 
		\label{eq:misc_n64}
	\end{equation}
	for all $u\in \calS(X)$ and $v\in L^2(X)$, where the left-hand side is defined as a distributional pairing: for all $v\in \calS'(X)$ and $u\in \calS(X)$, we write 
	\begin{equation}
		\langle u, Pv \rangle_{\calH } = Pv \Big( \frac{u^* \mathrm{d}\operatorname{Vol}_{g_X}}{(1+x\aleph+x^{1+\delta}\beth)^{1/2}} \Big), 
	\end{equation}
	where
	$Pv: \calS(X;|\Lambda^d| T^* X)\to \bbC$ is a tempered distribution. 
	
	In order to conclude that $P:\calS(X)\to L^2(X)$ is essentially self-adjoint with respect to the $L^2(X,(1+x\aleph+x^{1+\delta}\beth)^{-1/2} \dd \mathrm{Vol}_{g_X})$ inner product, it suffices to check that 
	\begin{equation}
		\overline{\operatorname{range}(P \pm i)} = L^2(X)
		\label{eq:misc_br4}
	\end{equation} 
	for both choices of sign \cite[Chp. VIII, \S2]{RS1}, where  $\operatorname{range}(P \pm i) = \{Pu\pm iu:u\in \calS(X)\}$.
	Let $\operatorname{ker}(P\mp i) = \{v\in \calS'(X) : Pu = \pm i u\}$. For all $v\in L^2(X) \subset \calS'(X)$, we have, via \cref{eq:misc_n64},  
	\begin{align}
		\begin{split} 
		v \in \operatorname{range}(P\pm i)^\perp &\iff \langle (P\pm i) u, v \rangle_\calH=0\text{ for all }u\in \calS(X) \\ 
		&\iff \langle  u, P v \rangle_\calH \mp \langle u,i v \rangle_{\calH}  = \langle  u, (P\mp i) v \rangle_\calH =0\text{ for all }u\in \calS(X) \\ 
		&\iff (P \mp i) v = 0 \Rightarrow v \in \operatorname{ker}(P\mp i).
		\end{split} 
	\end{align}
	So, $\operatorname{range}(P\pm i)^\perp \subseteq \operatorname{ker}(P\mp i)$. By elliptic regularity, $\operatorname{ker}(P\mp i)$ consists entirely of Schwartz functions.  Thus, if $v \in \operatorname{range}(P\pm i)^\perp$, then
	\begin{equation}
		0 = \langle v,(P\pm  i) v \rangle_{\calH}  = \langle v,Pv\rangle_{\calH} \pm i \lVert v \rVert_{\calH}^2. 
	\end{equation}
	Since the first term on the right-hand side is real by symmetry (using the fact that $v$ is Schwartz, so as to be able to appeal to the computations above), this forces $v=0$. So, $\operatorname{range}(P\pm i)^\perp=\{0\}$. Since 
	\begin{equation}
		\overline{\operatorname{range}(P \pm i)}= (\operatorname{range}(P\pm i)^\perp)^\perp,
	\end{equation}
	\cref{eq:misc_br4} follows. 
	
	We now know that $P:\calS(X)\to L^2(X)$ is essentially self-adjoint with respect to the $L^2(X,(1+x\aleph+x^{1+\delta}\beth)^{-1/2} \dd \mathrm{Vol}_{g_X})$ inner product. 
	Let 
	\begin{equation}
		\overline{P} : \calD(\overline{P}) \to L^2(X) 
	\end{equation}
	denote the closure of $P$. 
	It remains only to observe that $\calD(\overline{P}) = H^2(X)$ and that $\overline{P}u$ is the result of applying the differential operator $P$ to $u\in H^2(X)$. 
	\begin{itemize}
		\item Since $P\in \calL(H^2(X),L^2(X))$, any closure of $P:\calS(X)\to L^2(X)$ contains $H^2(X)$ in its domain and acts on this domain in the expected way. So, $\calD(\overline{P})\supseteq H^2(X)$, and $\overline{P}$ extends $P:H^2(X)\to L^2(X)$. 
		\item It can be shown that $P:H^2(X)\to L^2(X)$ is closed using the estimate 
		\begin{equation}
			\lVert  u \rVert_{H^2(X)} \preceq \lVert P u \rVert_{L^2(X)} + \lVert u \rVert_{H^1(X)}  +\lVert u \rVert_{L^2(X)} \preceq \lVert P u \rVert_{L^2(X)} + \lVert u \rVert_{L^2(X)}, 
			\label{eq:misc_029}
		\end{equation}
		where the second inequality is deduced from the first via the interpolation estimate 
		\begin{equation} 
		\lVert u \rVert_{H^1(X)} \preceq \lVert P u \rVert_{L^2(X)} + \lVert u \rVert_{L^2(X)},
		\end{equation} 
		and where `$a\preceq b$' denotes $a\leq Cb$ for some unspecified constant $C$ that can depend on the spacetime considered but not on the functions involved in the definitions of $a,b$.  
		If $\{u_n\}_{n=0}^\infty \subset H^2(X)$ satisfies $u_n \to u$ in $L^2(X)$ for some $u\in L^2(X)$, and if $Pu_n \to v$ in $L^2(X)$ for some $v\in L^2(X)$, then \cref{eq:misc_029} implies that $\{u_n\}_{n=0}^\infty$ is Cauchy in $H^2(X)$, and the $H^2$-limit is also an $L^2$-limit and therefore $u$, so 
		\begin{equation}
			u_n \to u \in H^2(X),
		\end{equation} 
		which also implies $v=P u$. 
	\end{itemize}
	Combining the previous two observations, we conclude that $\overline{P} = P$. 
	
	For all $u\in H^2(X)$, $\langle u,Pu \rangle_{\calH}$ is given by
	\begin{equation}
	\int_X \frac{u^*   P u \dd \mathrm{Vol}_{g_X} }{(1+x\aleph +x^{1+\delta}\beth)^{1/2}} = \int_X  \frac{u_0^* \tilde{P} u_0\dd \mathrm{Vol}_{g_X}}{1+x\aleph +x^{1+\delta}\beth}  = \langle u_0,\triangle u_0 \rangle_{L^2(X,\dd \mathrm{Vol}_{g_X}\!)} + \int_X \frac{W |u|^2 \dd \mathrm{Vol}_{g_X}}{(1+x\aleph +x^{1+\delta}\beth)^{1/2}}, 
	\end{equation}
	where $W$ is as in \cref{eq:W} and $u_0=(1+x\aleph + x^{1+\delta} \beth)^{1/4}u$. 
	From the semidefiniteness of $\triangle$ on $L^2(X,\dd \mathrm{Vol}_{g_X}\!)$, we conclude that $\langle u,Pu\rangle_\calH \geq (\inf W) \lVert u \rVert_{\calH}^2$. So, $P$ is lower-semibounded. 
\end{proof}

\begin{proposition}
	If $\aleph<0$, there exists some infinite sequence $\{v_n\}_{n=1}^\infty \subseteq C_{\mathrm{c}}^\infty(X^\circ)$ such that $\operatorname{supp} v_n \cap \operatorname{supp} v_{n'} = \varnothing$ if $n\neq n'$ and $\langle v_n, P v_n \rangle_{L^2(X,(1+x\aleph+x^{1+\delta}\beth)^{-1/2} \dd \mathrm{Vol}_{g_{ X}}) } < 0$ for all $n$.
	\label{prop:mainlemma}
\end{proposition}
\begin{proof}
	By \Cref{prop:squareform}, there exists some $Q_0\in S^{0} \operatorname{Diff}_{\mathrm{sc}}^2(X)$ such that
	\begin{equation} 
		P = - \Big( 1 - \frac{\mathsf{r}_0}{r} \Big) \frac{\partial^2}{\partial r^2} - \frac{d-1}{r} \frac{\partial}{\partial r}  + \frac{m^2 \aleph}{r}  + \Big(1+\frac{\aleph}{r}\Big) \frac{\triangle_{\partial X}}{r^2} + \frac{1}{r} Q  + \frac{1}{r^{1+\delta}}Q_0 
	\end{equation} 
	near $\partial X$, where $\mathsf{r}_0\in \bbR$ is defined by $\mathsf{r}_0 = \mathsf{M}- \aleph$. We will work with $v$ supported in the boundary collar, with respect to which we impose that $v$ depends only on $r$. Then, $\triangle_{\partial X}v,Qv=0$.
	Thus, 
	\begin{equation} 
	\langle v,Pv \rangle_\calH = \Big\langle v , - \Big( 1 - \frac{\mathsf{r}_0}{r} \Big) \frac{\partial^2v}{\partial r^2} - \frac{d-1}{r} \frac{\partial v}{\partial r}  + \frac{m^2 \aleph v}{r}  \Big\rangle_\calH +\Big\langle v, \frac{1}{r^{1+\delta}} Q_0 v \Big\rangle_\calH .
	\end{equation} 
	Since $\aleph<0$, this yields 
	\begin{equation}
	\langle v,Pv \rangle_\calH \leq - \frac{m^2 |\aleph|}{R_0} \lVert v \rVert_{\calH^2} +  \Big\langle v , - \Big( 1 - \frac{\mathsf{r}_0}{r} \Big) \frac{\partial^2v}{\partial r^2} - \frac{d-1}{r} \frac{\partial v}{\partial r}   \Big\rangle_\calH +\Big\langle v, \frac{1}{r^{1+\delta}} Q_0 v \Big\rangle_\calH
	\label{eq:misc_036}
	\end{equation}
	if $v$ is supported in $\{r\leq R_0\}$. 
	Fix nonzero $\chi \in C_{\mathrm{c}}^\infty(\bbR;\bbR^{\geq 0})$ with $\chi(0)=1$ and $\operatorname{supp} \chi \Subset (-1,\infty)$, and let $v[\lambda](r) = \lambda^{-d/2}\chi( (r - \lambda) / \lambda)$ for $\lambda\geq 1$. If $\lambda$ is sufficiently large, then this is supported in the boundary collar, and we can consider $v\in C_{\mathrm{c}}^\infty(X^\circ)$. Also, this is supported in $\{R\leq r\leq R_0\}$ for $R_0=O(\lambda)$ and $R=\Omega(\lambda)$, so \cref{eq:misc_036} applies. 
	
	We can write the density $(1+x\aleph + x^{1+\delta}\beth)^{-1/2} \dd \mathrm{Vol}_{g_{X}}$ near $\partial X$ as $(1+x\aleph + x^{1+\delta}\beth)^{-1/2} \dd \mathrm{Vol}_{g_{X}} \in r^{d-1}(1 +  S^{-1}(X)) \dd r \dd \mathrm{Vol}_{h_{\partial X}}$. 
	Thus, if $v(r)$ is supported in $(R,\infty)_r$ for $R$ sufficiently large, which we denote by $R\gg 0$, we can estimate 
	\begin{equation}
	(1-cR^{-1})\lVert r^{(d-1)/2} v \rVert^2_{L^2(R,\infty)}\leq \operatorname{Vol}_{h_{\partial X}}(\partial X)^{-1} \lVert v \rVert_{\calH}^2 \leq (1+CR^{-1})\lVert r^{(d-1)/2} v \rVert^2_{L^2(R,\infty)}
	\label{eq:fundbd}
	\end{equation}
	for some $c,C>0$. 
	For each $\kappa \in \bbN$, 
	\begin{equation} 
	\lambda^{2\kappa} \lVert r^{(d-1)/2}\partial_r^\kappa v[\lambda] \rVert_{L^2(R,\infty)}^2 = 
	\int_0^\infty \Big( \frac{r}{\lambda}\Big)^d \chi^{(\kappa)}\Big( \frac{r}{\lambda} - 1 \Big)^2  \frac{\dd r}{r} = \int_{0}^\infty  \rho^d \chi^{(\kappa)}(\rho-1) \frac{\dd \rho}{\rho}
	\end{equation} 
	is independent of $\lambda \gg 0$.
	So, from \cref{eq:misc_036} and Cauchy--Schwarz, we get $\langle v,P v \rangle_{\calH} \leq - c /\lambda + O(\lambda^{-1-\delta})$ for some other $c>0$. This is negative if $\lambda$ is sufficiently large. Taking a sequence of $\lambda_1,\lambda_2,\cdots$ sufficiently large, the supports of the $v[\lambda_n]$ are disjoint. 
\end{proof}

\begin{proposition}
	If $\aleph<0$, then there exists some infinite sequence $\{E_n\}_{n=1}^\infty \subset \bbR^+$ such that  $E_n\downarrow 0$ as $n\to\infty$ and such that there exist $L^2(X,(1+x\aleph+x^{1+\delta}\beth)^{-1/2} \dd \mathrm{Vol}_{g_{ X}})$-orthonormal $u_1,u_2,\cdots\in  \calS(X)$ such that $P(E_n) u_n = 0$. 
	\label{prop:infinitude}
\end{proposition}
\begin{proof}
	If $u\in \calS'(X)$ satisfies $P(E)u=0$ for $E>0$, then $u\in \calS(X)$, since $P(E) = P +E$ is an elliptic element of the sc-calculus on $X$. So, we need only construct $u_1,u_2,\cdots$ as elements of $L^2(X)$, and then they are automatically Schwartz.  

	Via analytic Fredholm theory, $\sigma(P) \cap (-\infty,0) = \sigma_{\mathrm{pp}}(P)\cap (-\infty,0)$, 
	and $\sigma_{\mathrm{pp}}(P)\cap (-\infty,0)$ has no accumulation points within $(-\infty,0)$. 
	So, $\sigma_{\mathrm{cont}}(P) \subseteq [0,\infty)$. 
	(In fact, equality holds: $\sigma_{\mathrm{cont}}(P) = [0,\infty)$.)  So (using the fact that $P$ is lower-semibounded), we can conclude the proposition from the claim that 
	$\sigma_{\mathrm{pp}}(P) \cap (-\infty,0)$ is infinite. 
	First, let 
	\begin{equation}
		\mu_n  = 
		\begin{cases} 
		\inf \{  \lVert v \rVert_{\calH}^{-2} \langle v ,Pv \rangle_{\calH} :v\in L^2(X)\backslash \{0\}\}  & (n=1),\\ 
		\operatorname{sup}\{ \inf \{  \lVert v \rVert_{\calH}^{-2} \langle v ,Pv \rangle_{\calH} :v\in \{\varphi_1,\cdots,\varphi_{n-1}\}^\perp\backslash \{0\}\}   : \varphi_1,\ldots,\varphi_{n-1} \in L^2(X)\} & (n\geq 2),
		\end{cases} 
	\end{equation}
	for each $n\in \bbN^+$, where the orthogonal complements here and below are taken in $\calH$. 
	
	From the previous proposition, 
	\begin{equation} 
		\mu_n \leq \max \{ \lVert v_{j} \rVert_{\calH}^{-2} \langle v_{j},Pv_{j}\rangle_{\calH} : j=1,\ldots,n \}< 0=\inf \sigma_{\mathrm{cont}}(P)
		\label{eq:misc_049}
	\end{equation} 
	for all $n$. 
	Indeed, any $\varphi_1,\ldots,\varphi_{n-1} \in L^2(X)$ have the form $\varphi_k = \phi_k+ \sum_{j=1}^n a_{j,k} v_j^\circ$ for $\phi_k \in \smash{\{v_1,\ldots,v_n\}^\perp}$, where $\smash{v_j^\circ} = v_j / \lVert v_j \rVert_\calH$ and $a_{j,k}\in \bbC$. Let $a_k \in \bbC^n$ be the vector with components $(a_{1,k},\ldots,a_{n,k})$. As the dimension of the span of $a_1,\ldots,a_{n-1}$ is at most $n-1$, there exists some vector $a_n = (a_{1,n},\dots,a_{n,n}) \in \bbC^n$ of norm $1$ orthogonal to all of $a_1,\ldots,a_{n-1}$. Set $v = \sum_{j=1}^n a_{j,n} v_j^\circ \in \calH$.
	Because the $v_j^\circ$ have disjoint support and are therefore orthogonal in $\calH$, $\lVert v \rVert_{\calH}=1$, and $\langle v,\varphi_k \rangle_\calH = \langle a_n,a_k \rangle_{\bbC^n} = 0$, so $v\in \smash{\{\varphi_1,\ldots,\varphi_{n-1}\}^\perp }\backslash \{0\}$. Finally, because $P$ is a differential operator and therefore \emph{local}, $\langle v_j, P v_k \rangle_{\calH}=0$ for all $j\neq k$, so that $\langle v,P v \rangle_\calH = \sum_{j=1}^n |a_{j,n}|^2 \langle v_j^\circ, Pv_j^\circ \rangle_{\calH}$.

	Via the min-max version of the variational principle \cite[Thm. XIII.1]{RS4}, we conclude from \cref{eq:misc_049} that there exist infinitely many negative eigenvalues of $P$. This is counted with multiplicity, so this does not rule out the possibility that $\sigma_{\mathrm{pp}}(P)\cap (-\infty,0)$ might be finite. However, via the ellipticity of $P+E$ for $E>0$, each negative eigenvalue in fact has finite multiplicity, so we can actually conclude that $\sigma_{\mathrm{pp}}(P) \cap (-\infty,0)$ is infinite. 
\end{proof}

Specific details aside, the previous argument is a version of \cite[Thm. XIII.6a]{RS4}.

Since $P(\sigma)$ has real coefficients, we may take $u_n$ to be $\bbR$-valued without loss of generality. 

Finally, via one last calculation, directly from \Cref{prop:squareform}:
\begin{proposition}
	If $u \in \calS(X)$ satisfies $P(E)u=0$ for some $E>0$, then the function $U:\bbR_t\times X \to \bbC$ given by 
	\begin{equation}
		U(t,-) = 
		\begin{cases}
			e^{\pm i t \sqrt{m^2 - E} } u & (E \leq m^2), \\
			e^{\pm t \sqrt{E-m^2}} u & (E \geq m^2),
		\end{cases}
	\end{equation}
	satisfies the Klein--Gordon--Schr\"odinger equation $(\square+ m^2 + V) U = 0$, for either choice of sign.
\end{proposition}

Thus, if $u\neq 0$, then, choosing the sign appropriately in the $E\geq m^2$ case, $U$ is a non-decaying solution to the Klein--Gordon--Schr\"odinger equation on $(\bbR_t\times X,g)$.  

\section{Gravitational quasimodes}
\label{sec:quasimodes}

The argument in the previous section gives little information on the eigenvalues of $P$, besides the fact that there are infinitely many. The proof shows that there are $\Omega(-\log E)$ many eigenvalues in $(-\infty,E)$, but this is far from sharp; compare with the hydrogen atom, for which there are $\smash{\Omega(E^{-1/2})}$ such energy levels, counted without multiplicity.

It is natural to try to refine the result using better test functions. This is the purpose of this section.
We take the hydrogen bound states as test functions, the ``quasimodes'' referred to in the section title.

Via this idea, we prove: 

\begin{proposition}
	For $\mathsf{Z} = - m^2 \aleph$ and $\mathsf{r}_0 = \mathsf{M}-\aleph$, let, for each $n\in \bbN^+$ such that $n^2 \geq - \mathsf{r}_0 \mathsf{Z}$,
	\begin{equation} 
	E_n = 
	\begin{cases}
	\frac{\mathsf{Z}^2}{4n^2} & (\mathsf{r}_0=0), \\ 
	\frac{1}{\mathsf{r}_0^2} \big( \mathsf{r}_0 \mathsf{Z} + 2n^2 - 2n\sqrt{ n^2 +  \mathsf{r}_0 \mathsf{Z}} \big) & (\mathsf{r}_0\neq 0).
	\end{cases}
	\end{equation} 
	There exists some $C>0$ such that, for each $n$ as above, there exists an eigenvalue of $P$ in an interval of size $\smash{C n^{-\nu }}$ centered at $-E_n$, where $\nu = \min\{2(1+\delta),3\}$.  
	\label{prop:infinitude2}
\end{proposition}
\begin{remark}
	For any compact $K\Subset \bbR_{\mathsf{r}_0}\times (0,\infty)_{\mathsf{Z}}$, $E_n = \mathsf{Z}^2/(4 n^2)  + O_K(\mathsf{r}_0\mathsf{Z}^3 /n^4)$ for all $(\mathsf{r}_0,\mathsf{Z})\in K$. In particular, $E_n\downarrow 0$ as $n\to\infty$, for each individual $\mathsf{r}_0,\mathsf{Z}$, at an $n^{-2}$ rate, so an interval of size $n^{-\nu}$ is small relative to $E_n$, and the existence of infinitely many bound states follows from the proposition.
\end{remark}
\begin{proof}
	If $E_n$ is an eigenvalue of $P$, then there is nothing to prove, so assume otherwise. 
	Suppose that $v\in C^\infty_{\mathrm{c}}(X^\circ)$ is supported in the boundary collar and only depends on $r$.  We will construct $v=v[E_n]$ such that $\lVert v \rVert_{\calH}^{-1} \lVert (P+E_n) v \rVert_{\calH} \preceq n^{-\nu}$. This can be rewritten in terms of the resolvent $R(-E_n) = (P+E_n)^{-1}$, which is a bounded self-adjoint map on $\calH$: 
	\begin{equation}
	 n^{\nu} \preceq \lVert R(-E_n) (P+E_n) v \rVert_{\calH} \lVert (P+E_n) v \rVert_{\calH}^{-1} \leq \lVert R(-E_n) \rVert_{\mathrm{Op}} = 1/d(-E_n,\sigma(P)) ,
	\end{equation}
	where $d(-E_n,\sigma(P))$ is the distance from $-E_n$ to the spectrum $\sigma(P)$ of $P$. Thus, 
	\begin{equation}
	d(-E_n,\sigma(P)) \preceq n^{-\nu}.
	\end{equation}
	There is therefore a point of the spectrum within distance $O(n^{-\nu})$ of $E_n$. Since $E_n\sim 1/n^2$, if $n$ is sufficiently large this has to be an eigenvalue rather than a point in the continuous spectrum $[0,\infty)$. (And, for $n$ bounded, one can just take $C$ sufficiently large to make the proposition hold.)
	
	Letting
	\begin{equation} 
	P_0(E) = -\Big(1- \frac{\mathsf{r}_0}{r}\Big) \frac{\partial^2}{\partial r^2}  - \Big(\frac{d-1}{r} + \frac{\mathsf{r}_0 (3-d)}{r^2}\Big) \frac{\partial}{\partial r} +E - \frac{\mathsf{Z}}{r}   - \frac{(d^2-4d+3)}{4r^2} + \frac{\mathsf{r}_0  (d^2-8 d+15)}{4r^3},
	\end{equation}
	$v$ will be chosen such that 
	\begin{equation} 
	\lVert r^{(d-1)/2}  v \rVert_{L^2(\mathsf{R}_0,\infty)}^{-1}\lVert r^{(d-1)/2} P_0(E_n) v \rVert_{L^2(\mathsf{R}_0,\infty)} \preceq n^{-3},
	\label{eq:misc_047}
	\end{equation} 
	where $\mathsf{R}_0 = \max\{0,\mathsf{r}_0\}$. In addition, $v$ will be supported in $\{r \geq \Omega(n^2)\}$.
	Let us verify that this suffices. We can write $P= P_0 + r^{-(1+\delta)} Q_2$ for some $Q_2 \in S^0 \operatorname{Diff}_{\mathrm{sc}}^2(X)$. Thus, 
	\begin{equation}
	\lVert (P+E_n) v \rVert_{\calH} \leq \lVert P_0(E_n) v \rVert_{\calH} + \lVert r^{-(1+\delta)} Q_2 v \rVert_{\calH}. 
	\end{equation}
	Since $Q_2$ is bounded as a map $H^2(X) \to L^2(X)$, we have 
	\begin{equation} 
	\lVert r^{-(1+\delta)} Q_2 v \rVert_\calH = O(n^{-(1+\delta)} \lVert r^{(d-1)/2} v \rVert_{H^2(R,\infty)}),
	\end{equation} 
	using \cref{eq:fundbd}.
	Despite $P_0(E)$ not being uniformly elliptic as $E\to 0$, we can elementarily bound 
	\begin{align} 
	\lVert r^{(d-1)/2} v \rVert_{H^2(R,\infty)}^2 &\preceq \lVert r^{(d-1)/2} P_0(E_n) v \rVert_{L^2(R,\infty)}^2 + \lVert r^{(d-1)/2} v \rVert_{H^1(R,\infty)}^2  \\
	&\preceq \lVert r^{(d-1)/2} P_0(E_n) v \rVert_{L^2(R,\infty)}^2 + \varepsilon \lVert r^{(d-1)/2} v \rVert_{H^2(R,\infty)}^2 + \varepsilon^{-1} \lVert r^{(d-1)/2} v \rVert_{L^2(R,\infty)}^2 
	\label{eq:misc_7zz}
	\end{align} 
	for any $\varepsilon>0$, where the constants are independent of $\varepsilon$, $R\gg 0$, and $n$. Taking $\varepsilon$ sufficiently small, we can absorb the second term on the right-hand side of \cref{eq:misc_7zz} into the left-hand side, yielding 
	\begin{equation}
	\lVert r^{(d-1)/2} v \rVert^2_{H^2(R,\infty)} \preceq  \lVert r^{(d-1)/2} P_0(E_n) v \rVert_{L^2(R,\infty)}^2 + \lVert r^{(d-1)/2} v \rVert_{L^2(R,\infty)}^2,
	\end{equation} 
	where now $\varepsilon$ has been fixed.
	Combining all of this, we have, using \cref{eq:fundbd} and \cref{eq:misc_047}, 
	\begin{equation}
	\frac{\lVert (P+E_n) v \rVert_{\calH}}{\lVert v \rVert_{\calH}}
	\preceq 
	\frac{\lVert r^{(d-1)/2} P_0(E_n) v \rVert_{L^2(\mathsf{R}_0,\infty)}}{\lVert r^{(d-1)/2}  v \rVert_{L^2(\mathsf{R}_0,\infty)}} + O \Big( \frac{1}{n^{2(1+\delta)}} \Big) = O \Big( \frac{1}{n^\nu}\Big),
	\end{equation}
	as desired. 
	
	The construction of $v[E_n]$ is as follows. 
	For $E\geq 0$, consider $P_0(E)$ 
	on $(\mathsf{R}_0,\infty)$. The essentially unique solution $u=u[E]$ to $P_0(E)u = 0$ decaying exponentially as $r\to\infty$ is given by 
	\begin{equation}
	u[E](r) = C r^{(3-d)/2} (r-\mathsf{r}_0)^{-1}  e^{-\sqrt{E}(r-\mathsf{r}_0)} U\Big(  - \frac{\mathsf{Z}- \mathsf{r}_0 E}{2E^{1/2}},0, 2 E^{1/2} (r-\mathsf{r}_0) \Big),
	\label{eq:misc_n67}
	\end{equation} 
	where $U(a,b,z)$ denotes Tricomi's confluent hypergeometric function and $C=C[E]$ is an arbitrary nonzero factor. 
	For any four positive real numbers $\lambda<\lambda_0<\Lambda_0<\Lambda$, fix a function $\chi \in C_{\mathrm{c}}^\infty(\bbR)$ that satisfies $\operatorname{supp} \chi \Subset (\lambda,\Lambda)$,  $\chi(\rho)=1$ for all $\rho\in [\lambda_0,\Lambda_0]$, and $0\leq \chi(\rho)\leq 1$ for all $\rho\in \bbR$. Let $v_\chi[E](r) = \chi(2\mathsf{Z} E^{-1} (r-\mathsf{R}_0)^{-1})u[E](r)$. The rest of this section will be devoted to the check that $v=v_\chi[E]$ satisfies \cref{eq:misc_047} when $E=E_n$.
\end{proof}

The quantity $E_n$ satisfies the quadratic equation $4 E_n n^2= (\mathsf{Z} - \mathsf{r}_0 E_n)^2$, which means that the $a$-parameter in $U(a,b,z)$ in \cref{eq:misc_n67} is $-n$.
By choosing $C[E_n]$ appropriately, we can arrange that the function $u$ defined by \cref{eq:misc_n67} is
\begin{equation} 
	u[E_n](r) =  r^{(3-d)/2} \psi_n(  n E_n^{1/2}(r-\mathsf{r}_0)),
	\label{eq:misc_b65}
\end{equation} 
where 
\begin{equation}
	\psi_n(r) =   \sqrt{\frac{1}{ \pi n^5} } e^{-  r/n} L_{n-1}^1\Big(\frac{2r}{n}\Big)
	\label{eq:psin}
\end{equation}
denotes the $n$th s-orbital \textit{hydrogen wavefunction} \cite[Chapter X]{LandauLifshitz}\cite[\S18.3]{Hall}. Here, $L^1_{n}(z) = (n!)^{-1} z^{-1} e^z \frac{\mathrm{d}^n}{\mathrm{d}z^n} (e^{-z} z^{n+1})$ is a generalized Laguerre polynomial. 

The coefficient in \cref{eq:misc_b65} has been chosen for later convenience. 
For all $k\in \bbN$, we have $r^{k/2}\psi_n(r) \in L^2(\bbR^{\geq 0}_r)$, and the normalization is such that 
\begin{equation}
	\int_0^\infty 4\pi r^2 \psi_n^2(r) \dd r = 1. 
	\label{eq:normalization}
\end{equation}
More generally,  for any $k\in \bbN^+$, 
\begin{equation}
	\int_0^\infty r^{k} \psi_n^2(r) \dd r = \frac{f_k(n^2)}{n^2} 
	\label{eq:misc_btg}
\end{equation}
for some polynomial $f_k$ of degree $k-1$ with positive leading coefficient, which can be proven using the \emph{Kramers--Pasternack}  \cite{Pasternack}\cite{Kramers} recurrence relation. This computation can be found in many references, including \cite[Eq. 2, 3]{Pasternack}.
Thus, for any polynomial $g(r)\in \bbR[r]$ of degree $k\geq 1$ with positive leading coefficient, 
\begin{equation} 
	c n^{2k-4} \leq \int_0^\infty g(r) \psi_n^2(r) \dd r \leq  C n^{2k-4}
	\label{eq:misc_cn3}
\end{equation} 
for $n$ sufficiently large, 
for some $g$-dependent $c,C>0$. For example, 
\begin{equation}
	\int_0^\infty 4\pi r \psi_n^2(r) \dd r = \frac{1}{n^2}, \qquad \int_0^\infty 4\pi  r^3 \psi_n^2(r) \dd r = \frac{3n^2}{2}. \label{eq:misc_btk}
\end{equation}
The $k=0$ case of \cref{eq:misc_btg} is degenerate, and requires a somewhat different argument, e.g.\ using \textit{Pasternack's inversion relation}, which is also from \cite{Pasternack}. The result is
\begin{equation} 
	\int_0^\infty 4\pi \psi_n^2(r) \dd r = \frac{1}{n^3}. \label{eq:misc_btl}
\end{equation} 
These moment formulas will be our main input to the calculations below, making up for partial analytical understanding of the operator $P_0(E)$ near $\{E=\mathsf{Z}/r\}$ in the $E\to 0$ limit. The key point is that the equations \cref{eq:misc_btg}, \cref{eq:misc_cn3}, \cref{eq:misc_btk}, \cref{eq:misc_btl} tell us (via Markov's inequality) something about the concentration of the probability measure $4\pi r^2\psi_n^2(r) \dd r$ in the limit where $n\to\infty$.

The wavefunctions $\psi_n$ for large $n$ are known as \textit{Rydberg states} in the physics literature, where they are used to model atomic and molecular electrons on the threshold of ionization. 
The $n\to\infty$ behavior of the generalized Laguerre polynomials $L_{n-1}^1$ appearing in \cref{eq:psin} is very well understood, and we could, in principle, use this to get very precise asymptotic statements about $\psi_n(r)$ in the Rydberg limit. However, as this is a bit technically involved, and since an elementary argument suffices for the application above, we only carry out the elementary argument here. We summarize the upshot of the more precise analysis in \Cref{rem:precise}, but the proof is omitted.

Set, for each $n\in \bbN^+$,  
\begin{align}
	\begin{split}
		w_\chi[E_n](r) &= \sqrt{4\pi} (nE_n^{1/2})^{3/2}(r-\mathsf{r}_0) r^{(d-3)/2} v_\chi[E_n]  \\ 
		&= \sqrt{4\pi} (nE_n^{1/2})^{3/2} (r-\mathsf{r}_0) \chi\Big( \frac{2\mathsf{Z}}{E_n} \frac{1}{r-\mathsf{R}_0} \Big) \psi_n(nE_n^{1/2}(r-\mathsf{r}_0)). 
	\end{split}
\end{align} 
For some $n_0=n_0(\mathsf{Z},\mathsf{r}_0)>0$, we have, for all $n\geq n_0$, estimates 
\begin{equation}
	\lVert w_\chi[E_n] \rVert_{L^2(\mathsf{r}_0,\infty)} \preceq_{\mathsf{Z},\mathsf{r}_0,n_0} \lVert r^{(d-1)/2} v_\chi[E_n] \rVert_{L^2(\mathsf{R}_0,\infty)} \preceq_{\mathsf{Z},\mathsf{r}_0,n_0} \lVert w_\chi[E_n] \rVert_{L^2(\mathsf{r}_0,\infty)}, 
\end{equation}
so estimating $\lVert r^{(d-1)/2} v_\chi[E] \rVert_{L^2(\mathsf{R}_0,\infty)}$ amounts to estimating $\lVert w_\chi[E] \rVert_{L^2(\mathsf{r}_0,\infty)}$. 

When $\chi$ is close to the indicator function $1_{[\lambda,\Lambda]}$ in a suitable norm and $n$ is large, then the quantity 
\begin{align} 
	\begin{split} 
		\lVert w_\chi[E_n] \rVert_{L^2(\mathsf{r}_0,\infty)}^2 &= \int_{\mathsf{r}_0}^\infty 4 \pi n^3 E_n^{3/2} (r-\mathsf{r}_0)^2 \chi\Big( \frac{2\mathsf{Z}}{E_n} \frac{1}{r-\mathsf{R}_0} \Big)^2\psi_n(nE_n^{1/2}(r-\mathsf{r}_0))^2 \dd r \\ 
		&= \int_{0}^\infty 4 \pi r^2   \chi\Big( \frac{2\mathsf{Z} n}{E_n^{1/2}} \frac{1}{r- n E_n^{1/2} (\mathsf{R}_0-\mathsf{r}_0)} \Big)^2\psi_n(r )^2 \dd r 
		\label{eq:misc_b53}
	\end{split} 
\end{align}
has, according to Born's rule, the following physical interpretation: it is (approximately) the probability that an electron in the s-orbital in the $n$th hydrogen shell appears in the annulus $\{n^2 \Lambda^{-1} < r < n^2 \lambda^{-1}\}$ when the electron's position is measured. This annulus scales quadratically with $n$.

\begin{figure}
	\includegraphics[width = .65\textwidth]{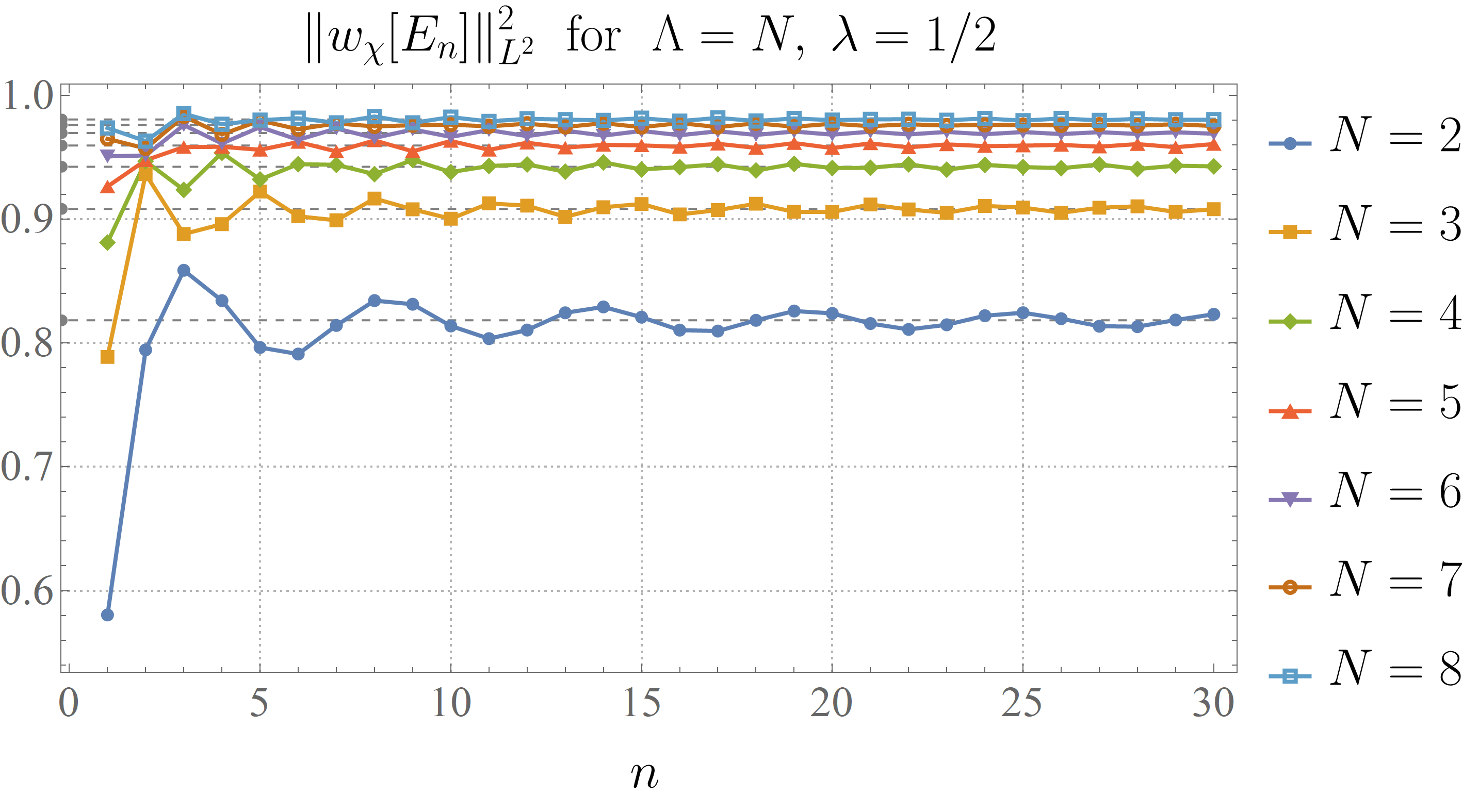}
	\caption{The $L^2$-norms $\smash{\lVert w_\chi [E_n]\rVert_{L^2}^2}$ of the cut-off hydrogen wavefunctions $w_\chi[E_n]$, with $\chi = 1_{[\lambda,\Lambda]}$ an indicator function, versus $n$. The function for the $7$ values $\Lambda\in \{2,\ldots,8\}$ are shown. The other parameters have been fixed at $\lambda=1/2$, $\mathsf{Z}=1$, and $\mathsf{r}_0=.8$. Dashed horizontal lines, marking the values of the $n\to\infty$ limits according to \Cref{rem:precise}, have been drawn at each of the $7$ vertical coordinates $2\pi^{-1}(N^{-1}(N-1)^{1/2} + \operatorname{arctan}((N-1)^{1/2}))$, for $N\in \{2,\ldots,8\}$.}
	\label{fig:main}
\end{figure}

\begin{proposition}
	Fix $\varepsilon >0$, and suppose that $\varepsilon \leq \Lambda_0$ and $\lambda_0\leq \varepsilon^{-1}$. 
	There exists some constant $C=C(\mathsf{Z},\mathsf{r}_0,\varepsilon)>0$, depending on $\mathsf{Z}$, $\mathsf{r}_0$, and $\varepsilon$, but nothing else, such that 
	\begin{equation} 
		1 - \frac{4}{\Lambda_0}- \frac{3\lambda_0 }{8} -  \frac{C}{n^2}  \leq \lVert w_\chi[E_n] \rVert_{L^2(\mathsf{r}_0,\infty)}^2   \leq  1
		\label{eq:cutoff_est}
	\end{equation}
	for all $n\in \bbN^+$. 
	\label{prop:cutoff_est}
\end{proposition}
\begin{proof}
	The upper bound in \cref{eq:cutoff_est} is just a consequence of \cref{eq:normalization}, \cref{eq:misc_b53}, and the assumption $\chi \leq 1$. In order to get the lower bound, we split
	\begin{equation}
		\lVert w_\chi [E_n] \rVert_{L^2(\mathsf{r}_0,\infty)}^2 = 1 -\int_0^\infty 4\pi r^2 \Big[1 -  \chi\Big( \frac{2\mathsf{Z} n}{E_n^{1/2}} \frac{1}{r- nE_n^{1/2} (\mathsf{R}_0-\mathsf{r}_0)} \Big)^2 \Big] \psi_n^2(r) \dd r.
	\end{equation}
	Since $\chi$ is identically equal to $1$ on $[\lambda_0,\Lambda_0]$, and since $\chi \leq 1$, 
	\begin{equation}
		\int_0^\infty 4\pi r^2 \Big[1 -  \chi\Big( \frac{2\mathsf{Z} n}{E_n^{1/2}} \frac{1}{r- nE_n^{1/2} (\mathsf{R}_0-\mathsf{r}_0)} \Big)^2 \Big] \psi_n^2(r) \dd r  \leq I_1 + I_2, 
	\end{equation}
	where 
	\begin{equation}
		I_1 = \int_0^{2\mathsf{Z} n  E_n^{-1/2} \Lambda_0^{-1} + n E_n^{1/2} (\mathsf{R}_0-\mathsf{r}_0)} 4\pi r^2 \psi_n^2(r) \dd r,\quad
		I_2 = \int_{2\mathsf{Z} n  E_n^{-1/2} \lambda^{-1}_0 + n E_n^{1/2} (\mathsf{R}_0-\mathsf{r}_0)}^\infty 4\pi r^2 \psi_n^2(r) \dd r. 
	\end{equation}
	We control these using two Markov bounds:
	\begin{align}
		I_1 &\leq \Big(\frac{2\mathsf{Z} n}{E_n^{1/2}\Lambda_0} + n E_n^{1/2} (\mathsf{R}_0-\mathsf{r}_0)\Big)\int_0^\infty 4\pi r \psi_n^2(r) \dd r = \frac{1}{n^2} \Big(\frac{2\mathsf{Z} n}{E_n^{1/2} \Lambda_0} + n E_n^{1/2} (\mathsf{R}_0-\mathsf{r}_0)\Big),\\
		I_2 &\leq \frac{\lambda_0}{2\mathsf{Z} n  E_n^{-1/2} +  n E_n^{1/2}\lambda_0 (\mathsf{R}_0-\mathsf{r}_0)}\int_0^\infty 4\pi r^3 \psi_n^2(r) \dd r \leq \frac{1}{2 }\frac{3n^2\lambda_0}{(2\mathsf{Z} n  E_n^{-1/2} + n E_n^{1/2} \lambda_0 (\mathsf{R}_0-\mathsf{r}_0))}. 
	\end{align}
	Combining these estimates, we get \cref{eq:cutoff_est}. 
\end{proof}

\begin{figure}
	\includegraphics[width = .675\textwidth]{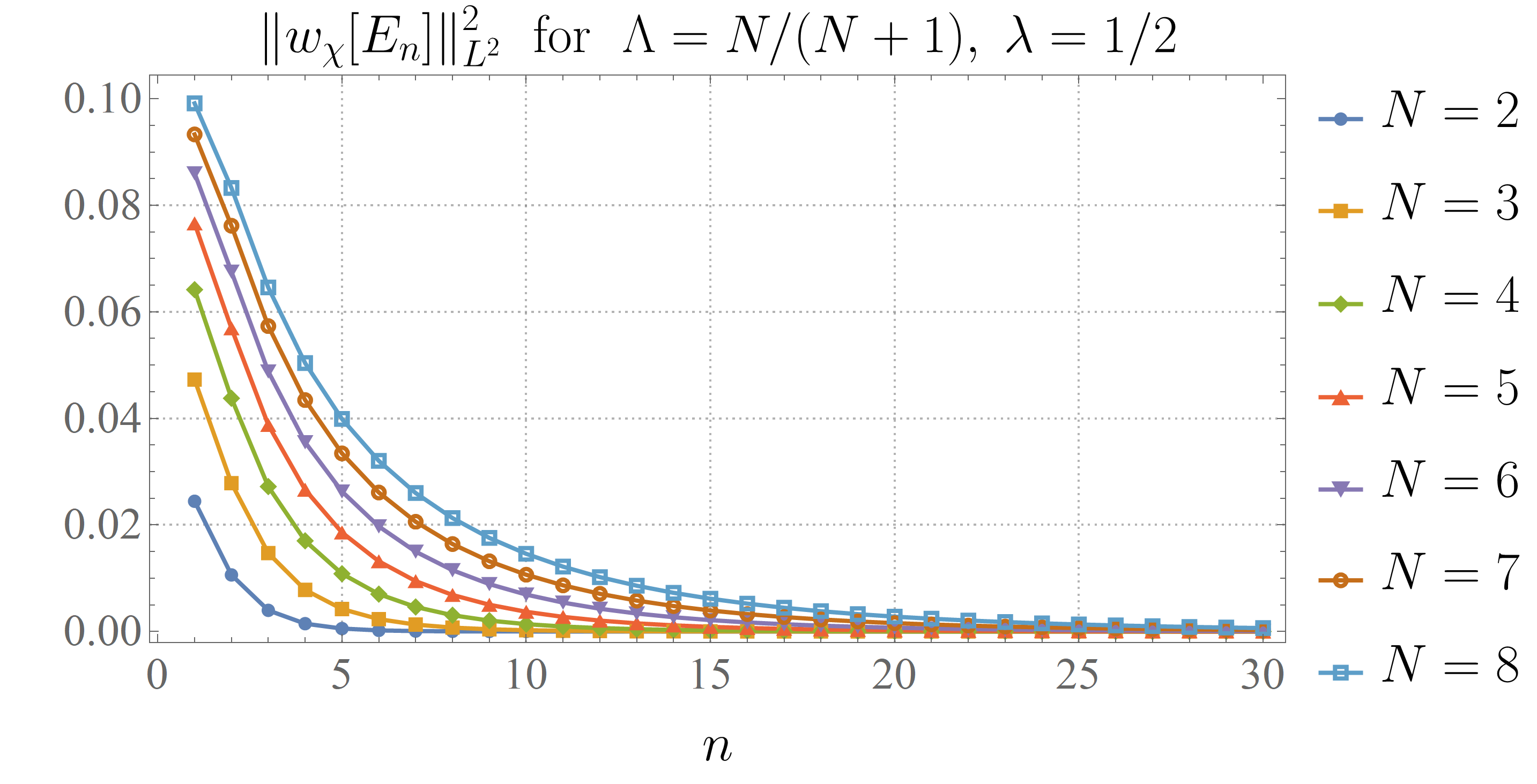}
	\caption{The quantities $\smash{\lVert w_\chi [E_n]\rVert_{L^2}^2}$ versus $n$, but now for $\Lambda = N/(N+1)<1$, $N\in \{2,\ldots,8\}$. In contrast to the situation with $\Lambda>1$, we see that the norms converge quickly (in fact, superpolynomially quickly, but we have not given the proof) to zero. Such values of $\Lambda$ are therefore unsuitable for the variational argument.}
\end{figure}

The key is, as long as $\Lambda_0$ is sufficiently large and $\lambda_0$ is sufficiently small, $\inf_{n\in \bbN^+} \lVert w_\chi[E_n] \rVert_{L^2(\mathsf{r}_0,\infty)} >0$.
The proof did not require that $\chi$ be differentiable; the same estimates (with $\lambda=\lambda_0$ and $\Lambda=\Lambda_0$) hold if $\chi= 1_{[\lambda,\Lambda]}$, and \Cref{fig:main} shows a plot of $\smash{\lVert w_\chi[E_n] \rVert_{L^2(\mathsf{r}_0,\infty)}^2}$ versus $n$ in this case.

\begin{remark}
	\label{rem:precise}
	For each $n\in \bbN^+$, let $\mu_n$ denote the probability measure on $[0,\infty)_{\hat{r}}$ whose density is given by $4\pi n^6 \hat{r}^2 \psi_n^2(n^2\hat{r}) \dd \hat{r}$.
	Using Erd\'elyi's uniform asymptotics for the Laguerre polynomials \cite{Erdelyi1, Erdelyi2}, it is possible to prove that these measures converge in law to the measure given by 
	\begin{equation}
		\mu_\infty(\hat{r}) = \frac{2}{\pi}   \Big(\frac{1}{\hat{r}} - 1 \Big)^{-1/2} 1_{\hat{r} \in [0,1]} \dd \hat{r}.
	\end{equation}
	It follows that $\lVert w_\chi[E_n] \rVert_{L^2(\mathsf{r}_0,\infty)}^2 \to \int_0^1 \chi(1/\hat{r}) \mu_\infty(\hat{r})$ as $n\to\infty$. By the portmanteau theorem \cite[Theorem 25.8]{Billingsley}, this holds even if $\chi = 1_{[\lambda,\Lambda]}$, as long as $1\notin \{\lambda,\Lambda\}$, in which case 
	\begin{align}
		\begin{split} 
			\lVert w_\chi[E_n] \rVert_{L^2(\mathsf{r}_0,\infty)}^2 &\to \frac{2}{\pi} \int_{\min\{1,1/\Lambda\}}^{\min\{1,1/\lambda\}} \Big(\frac{1}{\hat{r}} - 1 \Big)^{-1/2} \dd \hat{r} \\ 
			&= 
			\begin{cases}
				0 & (\Lambda<1), \\ 
				2\pi^{-1} [\Lambda^{-1}\sqrt{\Lambda-1 } + \operatorname{arctan}((\Lambda-1)^{1/2}) ]& (\Lambda>1>\lambda), \\
				2\pi^{-1} [L^{-1}\sqrt{L-1 } + \operatorname{arctan}((L-1)^{1/2})]^{L=\Lambda}_{L=\lambda} & (\lambda > 1),
			\end{cases}
		\end{split} 
		\label{eq:sharp_erd}
	\end{align}
	as $n\to\infty$. 
	Moreover, in the $\Lambda<1$ case, the decay to $0$ occurs at an exponential rate. For the values of $\lambda,\Lambda$ depicted in \Cref{fig:main}, we have marked the quantity on the right-hand side of \cref{eq:sharp_erd} via dashed horizontal lines.
\end{remark}

We also need to handle a derivative:
\begin{proposition}
	If $\lambda_0$ is sufficiently small and $\Lambda_0$ is sufficiently large, then 
	\begin{equation}
		\Big\lVert r^{(d-1)/2} \chi\Big(\frac{2\mathsf{Z}}{E_n (r-\mathsf{R}_0)} \Big) u'[E_n](r) \Big \rVert_{L^2(\mathsf{R}_0,\infty)}^2 \preceq_{d,\mathsf{Z},\mathsf{r}_0,\chi} \frac{1}{n^2} 
	\end{equation}
	for sufficiently large $n\in \bbN^+$, where the constant depends on $d,\mathsf{Z},\mathsf{r}_0,\chi$. 
	\label{prop:der_ctr}
\end{proposition}
\begin{proof}
	We have
	\begin{equation}
		\Big\lVert r^{(d-1)/2} \chi\Big( \frac{2\mathsf{Z}}{E_n (r-\mathsf{R}_0)} \Big) u'[E_n](r) \Big \rVert_{L^2(\mathsf{R}_0,\infty)}^2 = \int_{\mathsf{R}_0}^\infty \chi\Big( \frac{2\mathsf{Z}}{E_n (r-\mathsf{R}_0)} \Big)^2 u'[E_n](r)^2 r^{d-1} \dd r.  
	\end{equation}
	Integrating the right-hand side by parts, removing the derivative from one factor of $u'[E_n]$, yields $I_1+I_2+I_3$, where 
	\begin{align}
		\begin{split} 
			I_1 &= - (d-1) \int_{\mathsf{R}_0}^\infty \chi\Big( \frac{2\mathsf{Z}}{E_n (r-\mathsf{r}_0)} \Big)^2 u[E_n](r)u'[E_n](r) r^{d-2} \dd r,   \\
			I_2 &=  \frac{4\mathsf{Z}}{E_n }  \int_{\mathsf{R}_0}^\infty \frac{1}{(r-\mathsf{R}_0)^2}\chi'\Big( \frac{2\mathsf{Z}}{E_n (r-\mathsf{R}_0)} \Big) \chi\Big( \frac{2\mathsf{Z}}{E_n (r-\mathsf{R}_0)} \Big) u[E_n](r)u'[E_n](r)  r^{d-1} \dd r, \\
			I_3 &= - \int_{\mathsf{R}_0}^\infty \chi\Big( \frac{2\mathsf{Z}}{E_n (r-\mathsf{R}_0)} \Big)^2 u''[E_n](r) u[E_n](r) r^{d-1} \dd r.
		\end{split} 
	\end{align}
	We bound $I_1$, getting, for sufficiently large $n\in \bbN^+$,  
	\begin{align}
		\begin{split} 
			I_1 &\preceq_d \frac{1}{\varepsilon} \lVert r^{(d-3)/2} v_\chi[E_n] \rVert_{L^2(\mathsf{R}_0,\infty)}^2 +\varepsilon 	\Big\lVert r^{(d-1)/2} \chi\Big( \frac{2\mathsf{Z}}{E_n (r-\mathsf{R}_0)} \Big) u'[E_n](r) \Big \rVert_{L^2(\mathsf{R}_0,\infty)}^2 \\
			&\preceq_{d,\mathsf{Z},\mathsf{r}_0,\chi}  \frac{1}{\varepsilon n^4} +\varepsilon 	\Big\lVert r^{(d-1)/2} \chi\Big( \frac{2\mathsf{Z}}{E_n (r-\mathsf{R}_0)} \Big) u'[E_n](r) \Big \rVert_{L^2(\mathsf{R}_0,\infty)}^2,
		\end{split} 
	\end{align}
	for any $\varepsilon>0$, where the constants in the bounds do not depend on $\varepsilon$. 
	Similarly, for sufficiently large $n \in \bbN^+$, 
	\begin{align}
		\begin{split} 
			I_2 &\preceq_{\mathsf{Z},\mathsf{r}_0,\chi} \frac{1}{\varepsilon E_n^2} \lVert r^{(d-5)/2} v_{\bar{\chi}}[E_n] \rVert_{L^2(\mathsf{R}_0,\infty)}^2 +\varepsilon 	\Big\lVert r^{(d-1)/2} \chi\Big( \frac{2\mathsf{Z}}{E_n (r-\mathsf{R}_0)} \Big) u'[E_n](r) \Big \rVert_{L^2(\mathsf{R}_0,\infty)}^2 \\ 
			& \preceq_{d,\mathsf{Z},\mathsf{r}_0,\chi,\bar{\chi}} \frac{1}{\varepsilon n^4} + \varepsilon \Big\lVert r^{(d-1)/2} \chi\Big( \frac{2\mathsf{Z}}{E_n (r-\mathsf{R}_0)} \Big) u'[E_n](r) \Big \rVert_{L^2(\mathsf{R}_0,\infty)}^2, 
		\end{split} 
	\end{align}
	where we have fixed $\bar{\chi}\in C_{\mathrm{c}}^\infty((0,\infty);[0,1])$ that is identically equal to $1$ on the support of $\chi$. 
	Finally, 
	\begin{align}
		\begin{split} 
			I_3 &\preceq  \frac{1}{\varepsilon}  \lVert  r^{(d-2)/2} v_\chi[E_n] \rVert_{L^2(\mathsf{R}_0,\infty)}^2 + \varepsilon  \Big\lVert r^{d/2} \chi\Big( \frac{2\mathsf{Z}}{E_n (r-\mathsf{R}_0)} \Big) u''[E_n](r) \Big \rVert_{L^2(\mathsf{R}_0,\infty)}^2 \\
			& \preceq_{d,\mathsf{Z},\mathsf{r}_0,\chi} \frac{1}{\varepsilon n^2} + \varepsilon  \Big\lVert r^{d/2} \chi\Big( \frac{2\mathsf{Z}}{E_n (r-\mathsf{R}_0)} \Big) u''[E_n](r) \Big \rVert_{L^2(\mathsf{R}_0,\infty)}^2. 
		\end{split}
	\end{align}
	In order to bound the last term, we use the ODE $P_0(E_n)u[E_n]=0$:
	\begin{multline}
		\Big\lVert r^{d/2} \chi\Big( \frac{2\mathsf{Z}}{E_n (r-\mathsf{R}_0)} \Big) u''[E_n](r) \Big \rVert_{L^2(\mathsf{R}_0,\infty)}^2 \preceq_{d,\mathsf{Z},\mathsf{r}_0,\chi} \\ \frac{1}{n^4} \lVert r^{d/2} v_\chi[E_n] \rVert_{L^2(\mathsf{R}_0,\infty)}^2  + \lVert r^{(d-2)/2} v_\chi[E_n] \rVert_{L^2(\mathsf{R}_0,\infty)}^2 
		+ \Big\lVert r^{(d-2)/2} \chi\Big( \frac{2\mathsf{Z}}{E_n (r-\mathsf{R}_0)} \Big) u'[E_n](r) \Big \rVert_{L^2(r,\mathsf{R}_0)}^2  \\
		\preceq_{d,\mathsf{Z},\mathsf{r}_0,\chi} \Big\lVert r^{(d-1)/2} \chi\Big( \frac{2\mathsf{Z}}{E_n (r-\mathsf{R}_0)} \Big) u'[E_n](r) \Big \rVert_{L^2(r,\mathsf{R}_0)}^2 + \frac{1}{n^2}
	\end{multline}
	for sufficiently large $n \in \bbN^+$. 
	Combining the estimates above, 
	\begin{multline}
		\Big\lVert r^{(d-1)/2} \chi\Big( \frac{2\mathsf{Z}}{E_n (r-\mathsf{R}_0)} \Big) u'[E_n](r) \Big \rVert_{L^2(\mathsf{R}_0,\infty)}^2 \preceq_{d,\mathsf{Z},\mathsf{r}_0,\chi} \Big(\frac{1}{\varepsilon} +\varepsilon\Big)\frac{1}{n^2}  \\ + \varepsilon \Big\lVert r^{(d-1)/2} \chi\Big( \frac{2\mathsf{Z}}{E_n (r-\mathsf{R}_0)} \Big) u'[E_n](r) \Big \rVert_{L^2(\mathsf{R}_0,\infty)}^2, 
	\end{multline}
	where the constant in the bound  is independent of $\varepsilon$. Taking $\varepsilon$ sufficiently small, we can absorb the final term on the right-hand side into the left-hand side to conclude the result.
\end{proof}

\begin{proposition}
	Given the setup above, there exists some $C>0$ (depending on $d,\mathsf{Z},\mathsf{r}_0,\chi$ and nothing else) such that $\lVert r^{(d-1)/2} P_0(E_n) v_\chi[E_n] \rVert_{L^2}^2 \leq C n^{-6}$ for sufficiently large $n \in \bbN^+$. 
\end{proposition}
\begin{proof}
	We write $P_0(E_n) v_\chi[E_n]= v_1+v_2+v_3$, where
	\begin{align}
		\begin{split}
			v_1 &= - \Big(1 - \frac{\mathsf{r}_0}{r} \Big) \Big[ \Big(\frac{4\mathsf{Z}^2}{E^2_n} \frac{1}{(r-\mathsf{R}_0)^4} \chi'' \Big( \frac{2\mathsf{Z}}{E_n}\frac{1}{(r-\mathsf{R}_0)} \Big)  + \frac{4\mathsf{Z}}{E_n} \frac{1}{(r-\mathsf{R}_0)^3} \chi' \Big( \frac{2\mathsf{Z}}{E_n} \frac{1}{(r-\mathsf{R}_0)} \Big) \Big) \Big]u[E_n](r),  \\
			v_2 &= \frac{2\mathsf{Z}}{E_n} \frac{1}{(r-\mathsf{R}_0)^2}\Big( \frac{d-1}{r} +\frac{\mathsf{r}_0(3-d)}{r^2} \Big) \chi' \Big( \frac{2\mathsf{Z}}{E_n} \frac{1}{(r-\mathsf{R}_0)} \Big)  u[E_n](r), \\
			v_3 &= 2\Big(1 - \frac{\mathsf{r}_0}{r} \Big) \Big[ \frac{2\mathsf{Z}}{E_n} \frac{1}{(r-\mathsf{R}_0)^2} \chi' \Big( \frac{2\mathsf{Z}}{E_n} \frac{1}{(r-\mathsf{R}_0)} \Big)\Big] u'[E_n](r). \\
		\end{split}
	\end{align}
	For sufficiently large $n$, we can bound, via the estimates above, 
	\begin{align}
		\begin{split} 
		\lVert r^{(d-1)/2} v_1 \rVert_{L^2(\mathsf{R}_0,\infty)}^2, \lVert r^{(d-1)/2} v_2 \rVert_{L^2(\mathsf{R}_0,\infty)}^2 &\preceq_{d,\mathsf{Z},\mathsf{r}_0,\chi} \frac{1}{n^8} \lVert r^{(d-1)/2} v_{\bar{\chi}}[E_n] \rVert_{L^2(\mathsf{R}_0,\infty)}^2   \\
		&\preceq_{d,\mathsf{Z},\mathsf{r}_0,\chi} \frac{1}{n^8}, 
		\end{split} \label{eq:misc_n00} \\
		\begin{split} 
			\lVert r^{(d-1)/2} v_3 \rVert_{L^2(\mathsf{R}_0,\infty)}^2 &\preceq_{d,\mathsf{Z},\mathsf{r}_0,\chi} \frac{1}{n^4} \Big\lVert r^{(d-1)/2} \bar{\chi}\Big(\frac{2\mathsf{Z}}{E_n(r-\mathsf{R}_0)}\Big) u'[E_n] \Big\rVert_{L^2(\mathsf{R}_0,\infty)}^2 \\
			&\preceq_{d,\mathsf{Z},\mathsf{r}_0,\chi,\bar{\chi}} \frac{1}{n^6},
		\end{split} \label{eq:misc_n01}
	\end{align} 
	where $\bar{\chi}$ is as in the proof of the previous proposition. To get the last estimate, we applied \Cref{prop:der_ctr} with $\bar{\chi}$ in place of $\chi$. 
	Combining \cref{eq:misc_n00} and \cref{eq:misc_n01}, we arrive at the conclusion of this proposition. 
\end{proof}

Combining the propositions in this section, we get the estimate, \cref{eq:misc_047}, needed previously.

\printbibliography

\end{document}